\documentclass[12pt]{article}
\usepackage{amsfonts}
\usepackage{amssymb}
\usepackage{amsmath}
\usepackage{theorem}

\newtheorem{theo}{Theorem}[section]
{\theorembodyfont{\rmfamily}
\newtheorem{defin}[theo]{Definition}
}

\newenvironment{proof}{\textit{Proof.}}{\hfill$\square$}

{\theorembodyfont{\rmfamily}

}
{\theorembodyfont{\rmfamily}

}
{\theorembodyfont{\rmfamily}

}

\begin{document}

\title{On egalitarian values for cooperative games with a priori unions}
\author{J.M. Alonso-Meijide$^1$, J. Costa$^2$, \\
I. Garc\'{\i}a-Jurado$^3$, J.C. Gon\c{c}alves-Dosantos$^3$}
\date{\empty}
\maketitle

\begin{abstract}
In this paper we extend the equal division and the equal surplus division 
values for transferable utility cooperative games  to the more general setup of  transferable utility cooperative games with a priori 
unions. In the case of the equal surplus division value we propose 
three possible extensions. 
We provide axiomatic characterizations of the new values.
Furthermore, we apply the proposed modifications to a particular cost
sharing problem and compare the numerical results with those obtained with
the original values.
\end{abstract}

\footnotetext[1]{%
Grupo MODESTYA, Departamento de Estat\'{\i}stica, An\'alise Matem\'atica e
Optimizaci\'on, Universidade de Santiago de Compostela, Facultade de
Ciencias, Campus de Lugo, 27002 Lugo, Spain.} \footnotetext[2]{%
Grupo MODES, Departamento de Matem\'aticas, Universidade da Coru\~{n}a,
Campus de Elvi\~{n}a, 15071 A Coru\~{n}a, Spain.}
\footnotetext[3]{Grupo MODES, CITIC and Departamento de Matem\'aticas, Universidade da Coru\~{n}a, Campus de Elvi\~{n}a, 15071 A Coru\~{n}a, Spain.}

\noindent \textbf{Keywords:} cooperative games, a priori unions, equal
division value, equal surplus division value.

\section{Introduction}

Many economic problems deal with situations in which several agents
cooperate to generate benefits or to reduce costs. Cooperative game 
theory studies procedures to allocate the resulting benefits (or costs) 
among the cooperating agents in those situations.

One of the most commonly used allocating procedures is the Shapley value,
introduced in Shapley (1953) and analyzed more recently in Moretti and
Patrone (2008) or in Alonso-Meijide et al. (2019). Very often, however,
agents cooperate on the basis of a kind of egalitarian principle according
to which the benefits will be shared equitably. For instance, Selten (1972)
indicates that egalitarian considerations explain in a successful way
observed outcomes in experimental cooperative games.

In recent years, the game theoretical literature has dealt with several
egalitarian solutions in cooperative games. For instance, van den Brink
(2007) provides a comparison of the equal division value and the Shapley
value, and Casajus and H\"{u}ttner (2014) compare those two solutions with
the equal surplus division value (studied first in Driessen and Funaki, 1991). 
In van den Brink and Funaki (2009), Chun
and Park (2012), van den Brink et al. (2016), Ferri\`{e}res (2017) and B\'{e}%
al et al. (2019) several axiomatic characterizations of the equal division
and equal surplus division values are provided. Ju et al. (2007) introduce
and characterize the consensus value, a new solution that somewhat combines
the Shapley value and the equal division rule. Dutta and Ray (1989)
introduce the egalitarian solution for cooperative games, closely related to Lorenz dominance, that considers
cooperating agents who believe in equality as a desirable social goal and
negotiate accordingly; this solution was later characterized by Dutta
(1990), Klijn et al. (2000) and Ar\'{\i}n et al. (2003), and modified by
Dietzenbacher et al. (2017).

Another stream of literature in cooperative game theory started in Owen
(1977), where a variant of the Shapley value for games with a priori unions
is introduced and characterized. In a game with a priori unions there exists
a partition of the set of players, whose classes are called unions, that is
interpreted as an a priori coalition structure that conditions the
negotiation among the players and, consequently, modifies the fair outcome
of the negotiation. There is a large literature concerning the Owen value
and its applications; just to cite some recent papers, Lorenzo-Freire (2016)
provides new axiomatic characterizations of the Owen value, Costa (2016)
deals with an application in a cost allocation problem, and Saavedra-Nieves
et al. (2018) propose a sampling procedure to approximate it. Not only the
Shapley value but also other values have been modified for the case with a
priori unions. For instance, Alonso-Meijide and Fiestras-Janeiro (2002) deal
with the Banzhaf value for games with a priori unions, Casas-M\'{e}ndez et
al. (2003) introduce the $\tau $-value for games with a priori unions,
Alonso-Meijide et al. (2011) study the Deegan-Packel index for simple games
with a priori unions, and Hu et al. (2019) introduce an egalitarian
efficient extension of the Aumann-Dr\`{e}ze value (Aumann and Dr\`{e}ze,
1974).

In this paper we modify the equal division value and the equal surplus
division value for games with a priori unions. In Section 2 we illustrate
the interest of our study describing a cost allocation problem that arises
in the installation of an elevator in an apartment building. In Section 3 we
define and characterize the equal division rule for games with a priori
unions. In Section 4 we introduce and characterize three alternative
extensions of the equal surplus division rule for games with a priori
unions. In Section 5 we include some final remarks.

\section{An example}

In this section we consider an example where the owners of apartments in a
building have agreed to install an elevator and share the corresponding
costs. This example is inspired by a problem analyzed in Crettez and Deloche
(2018) from the point of view of French legislation. The French Law on
Apartment Ownership of Buildings does not provide a precise method for
sharing the cost of an improvement but indicates that the co-owners must
pay \textquotedblleft in proportion to the advantages\textquotedblright\
they will receive. In the case of elevators in France, Crettez and Deloche
(2018) indicate that there is a de facto sharing method that they call the 
\emph{elevator rule}. In their paper they study the elevator rule and other
proposals in the spirit of the French legislation.

However, Crettez and Deloche (2018) explain that in other European countries
the legislation is based on principles of egalitarian character. For
example, in The Netherlands each of the owners of the apartments must
\textquotedblleft participate for an equal part in the debts and costs which
are for account of all apartments owned pursuant to law or the internal
arrangements, unless the internal arrangements provide for another
proportion of participation.\textquotedblright

The Spanish Horizontal Property Law 49/1960 (modified by the Act 8/2013)
indicates that \textquotedblleft to each apartment or local will be attributed a
quota of participation in relation to the total of the value of the building
(\dots ). This quota will serve as a module to determine the participation
in the burdens and benefits due to the community.\textquotedblright\ These
quotas generally depend on the surface area of each apartment but can take
into account other aspects.

In a particular example, let us see how the Dutch and Spanish rules would
share the costs of installing an elevator. Consider the following
three-storey building with no apartments or offices on the ground floor: on the first floor there is a single apartment of 180
square meters, on the second floor there are two apartments, one of 100 and
other one of 90 square meters, and on the third floor there are three
apartments of 60 square meters each. The second floor has a slightly larger
area because one of the two apartments on the floor has an additional
gallery. Suppose now that the cost of installing the elevator is 120 (in
thousands of euros), 50 of which correspond to the machine, 40 to the works
to make the hollow of the elevator, and 30 to the works to be done on each
floor to allow access to the elevator (10 in each of them). Table \ref%
{table1} below shows the distribution of costs for each of the apartments
according to the Dutch and Spanish rules (the latter with quotas for each
apartment given by its surface). Notice that both rules are based on
egalitarian principles and can be interpreted as the equal division rule;
the difference is that in the case of the Dutch rule the subjects that
receive the equitable distribution are the apartments, whereas in the case
of the Spanish rule the equitable distribution subjects are the quota units.%
\footnote{%
In this example the quota units are the square meters of the apartments. For
the approach we adopt to be meaningful, the quota unit numbers must be
integers.} Notice that the same egalitarian spirit of these rules can be
maintained despite changing the equitable distribution subjects. For
instance, it would be natural to consider a kind of two-step equitable
distribution subjects, where the subjects in the first step are the floors
and the subjects in the second step are the apartments (in the case of the
Dutch rule) or the quota units (in the case of the Spanish rule). This would
result in the distribution of costs shown in Table \ref{table2} below.
Observe that this variation arises from considering that the floors of the
building naturally give rise to a structure of a priori unions in the sense
of Owen (1977) and, thus, the convenience of extending the equal division
value for games with a priori unions emerges spontaneously in this example.
We do it formally in Section 3.

There are other possible variations of these Dutch and Spanish rules with
and without the structure of a priori unions when using the equal surplus
division value instead of the equal division value; thus, the convenience of
extending the equal surplus division value for games with a priori unions
can also be motivated on the basis of this example. We do it in Section 4,
where we also analyse in more depth how the equal surplus division value for
games with a priori unions can be applied in the example we have discussed
in this section.

\vspace*{0.5cm}

\begin{table}[htbp]
\begin{center}
\begin{tabular}{|l|l|l|}
\hline
& Dutch rule & Spanish rule \\ \hline
3rd floor & 20 20 20 & 13.0909 13.0909 13.0909 \\ \hline
2nd floor & 20 20 & 21.8182 19.6364 \\ \hline
1st floor & 20 & 39.2727 \\ \hline
\end{tabular}%
\end{center}
\caption{Distribution according to the Dutch and Spanish rules}
\label{table1}
\end{table}

\begin{table}[htbp]
\begin{center}
\begin{tabular}{|l|l|l|}
\hline
& Dutch rule & Spanish rule \\ \hline
3rd floor & 13.3333 13.3333 13.3333 & 13.3333 13.3333 13.3333 \\ \hline
2nd floor & 20 20 & 21.0526 18.9474 \\ \hline
1st floor & 40 & 40 \\ \hline
\end{tabular}%
\end{center}
\caption{Distribution according to the two-step Dutch and Spanish rules}
\label{table2}
\end{table}

\section{The equal division value for TU-games with a priori unions}

In this section we extend the equal division value for TU-games to the more
general setup of TU-games with a priori unions. To start with, we introduce
the basic concepts and notations we use in this paper.

A transferable utility cooperative game (from now on a {TU-game}) is a pair $%
(N,v)$ where $N$ is a finite set of $n$ players, and $v$ is a map from $2^N$
to $\mathbb{R}$ with $v(\emptyset)=0$, that is called the characteristic
function of the game. In the sequel, $\mathcal{G}_N$ will denote the family
of all TU-games with player set $N$ and $\mathcal{G}$ the family of all
TU-games. A {value} for TU-games is a map $f$ that assigns to every TU-game $%
(N,v)\in\mathcal{G}$ a vector $f(N,v)=(f_i(N,v))_{i\in N}\in\mathbb{R}^N$
with $\sum_{i\in N}f_i(N,v)=v(N)$.

As it was remarked in the introduction, sometimes agents cooperate on the
basis of a kind of egalitarian principle according to which the benefits
will be shared equitably. This gives rise to the {\ equal division value} $%
ED $ that distributes $v(N)$ equally among the players in $N$. Formally, the
equal division value $ED$ is defined for every $(N,v)\in \mathcal{G}$ and
for all $i\in N$\ by 
\begin{equation*}
ED_{i}(N,v)=\frac{v(N)}{n}.
\end{equation*}

Now denote by $P(N)$ the set of all partitions of $N$. A {TU-game with a
priori unions} is a triplet $(N,v,P)$ where $(N,v)\in\mathcal{G}$ and $%
P=\{P_1,\dots,P_m\}\in P(N)$. The set of TU-games with a priori unions and
with player set $N$ will be denoted by $\mathcal{G}_N^{U}$, and the set of
all TU-games with a priori unions by $\mathcal{G}^{U}$. A {value for
TU-games with a priori unions} is a map $g$ that assigns to every $(N,v,P)\in%
\mathcal{G}^U$ a vector $g(N,v,P)=(g_i(N,v,P))_{i\in N}\in\mathbb{R}^N$ with 
$\sum_{i\in N}g_i(N,v,P)=v(N)$. The next definition provides the natural
extension of the equal division value to TU-games with a priori unions.

\begin{defin}
The equal division value for TU-games with a priori unions $ED^{U}$ is
defined by 
\begin{equation*}
ED_{i}^{U}\left( N,v,P\right) =\frac{v(N)}{mp_{k}}
\end{equation*}%
for all $i\in N$ and all $(N,v,P)\in \mathcal{G}^{U}$ with $P=\{P_{1},\dots
,P_{m}\}$ and $i\in P_{k}$; $p_{k}$ denotes the cardinal of $P_{k}$.
\end{defin}

Notice that the equal division value for TU-games with a priori unions has
been used in the motivating example in Section 2 (see Table \ref{table2} and
the corresponding comments). Next we provide an axiomatic characterization
of this value. We start giving some properties of a value $g$ for TU-games
with a priori unions.

\bigskip \noindent \textbf{Additivity (ADD).} A value $g$ for TU-games with
a priori unions satisfies additivity if, for all $(N,v,P), (N,w,P)\in%
\mathcal{G}^U$, it holds that 
\begin{equation*}
g(N,v+w,P)=g(N,v,P)+g(N,w,P).
\end{equation*}

Take a TU-game $(N,v)\in \mathcal{G}^{N}$ and $i,j\in N$. We say that $i,j$
are indistinguishable in $v$ if $v(S\cup i)=v(S\cup j)$ for all $S\subseteq
N\setminus \{i,j\}$.

\bigskip \noindent \textbf{Symmetry within unions (SWU).} A value $g$ for
TU-games with a priori unions satisfies symmetry within unions if, for all $%
(N,v,P)\in\mathcal{G}^U$, all $P_k\in P$, and all $i,j\in P_k$
indistinguishable in $v$, it holds that $g_i(N,v,P)=g_j(N,v,P)$.

\bigskip Take $(N,v,P)\in\mathcal{G}^{U}$ with $P=\{P_1,\dots,P_m\}$ and
denote $M=\{1,\dots,m\}$. The {quotient game} of $(N,v,P)$ is the TU-game $%
(M,v/P)$ where 
\begin{equation*}
(v/P)(R)=v(\cup_{r\in R}P_r)\quad\mbox{for all}\;R\subseteq M.
\end{equation*}

\noindent \textbf{Symmetry among unions (SAU).} A value $g$ for TU-games
with a priori unions satisfies symmetry among unions if, for all $(N,v,P)\in%
\mathcal{G}^U$ and all $k,l\in M$ indistinguishable in $v/P$, it holds that $%
\sum_{i\in P_k}g_i(N,v,P)=\sum_{i\in P_l}g_i(N,v,P)$.

\bigskip Take a TU-game $(N,v)\in \mathcal{G}^N$ and $i\in N$. We say that $%
i $ is a nullifying player in $v$ if $v(S\cup i)=0$ for all $S\subseteq N$.

\bigskip \noindent \textbf{Nullifying player property (NPP).} A value $g$
for TU-games with a priori unions satisfies the nullifying player property
if, for all $(N,v,P)\in \mathcal{G}^{U}$ and all $i\in N$ nullifying player
in $v,$ it holds that $g_{i}(N,v,P)=0$.

\bigskip An analogous to NPP above is used in van den Brink (2007) to
characterize the equal division value for TU-games. In the next theorem, we
extend van den Brink's result to TU-games with a priori unions.

\begin{theo}
$ED^U$ is the unique value for TU-games with a priori unions that satisfies
ADD, SWU, SAU and NPP. \label{th2}
\end{theo}

\begin{proof} It is immediate to check that $ED^{U}$ satisfies ADD, SWU, SAU and NPP. To
prove the unicity, consider a value $g$ for TU-games with a priori unions
that satisfies ADD, SWU, SAU and NPP. Fix $N$ and define for all $\alpha \in 
\mathbb{R}$ and all non-empty $T\subseteq N$ the TU-game $(N,e_{T}^{\alpha
}) $ given by $e_{T}^{\alpha }(S)=\alpha $ if $S=T$ and $e_{T}^{\alpha
}(S)=0 $ if $S\neq T$. If $T=N$, since $g$ satisfies SWU and SAU, it is
clear that $g_{i}(N,e_{N}^{\alpha },P)=\frac{\alpha }{mp_{k}}$ for any $%
P=\{P_{1},\dots ,P_{m}\}$ and all $i\in P_{k}\subseteq N$, because all
players in $N$ are indistinguishable in $e_{N}^{\alpha }$ and all players in 
$M$ are indistinguishable in $e_{N}^{\alpha }/P$. If $T\subset N$ notice
that all players in $N\setminus T$ are nullifying players in $e_{T}^{\alpha
} $ and then, since $g$ satisfies NPP, 
\begin{equation*}
\sum_{i\in T}g_{i}(N,e_{T}^{\alpha },P)=\sum_{i\in N}g_{i}(N,e_{T}^{\alpha
},P)=e_{T}^{\alpha }(N)=0
\end{equation*}%
for any $P$. Then, since $g$ satisfies SWU and SAU it is not difficult to check that 
$g(N,e_{T}^{\alpha },P)=0$. Finally, the additivity of $g$ and the fact that $%
v=\sum_{T\subseteq N}e_{T}^{v(T)}$ imply that 
\begin{equation*}
g_{i}(N,v,P)=\sum_{T\subseteq
N}g_{i}(N,e_{T}^{v(T)},P)=g_{i}(N,e_{N}^{v(N)},P)=\frac{v(N)}{mp_{k}}%
=ED_{i}^{U}(N,v,P)
\end{equation*}%
for any $P$ and all $i\in P_{k}\subseteq N$.
\end{proof}

\section{The equal surplus division value for TU-games with a priori unions}
\label{sec_esdv}

In this section we extend the equal surplus division value for TU-games to
the more general setup of TU-games with a priori unions. To start with,
remember that the equal surplus division value $ESD$ is defined for every $%
(N,v)\in \mathcal{G}$ and for all $i\in N$\ by 
\begin{equation*}
ESD_{i}(N,v)=v(i)+\frac{v^{0}(N)}{n},
\end{equation*}%
where $v^{0}(S)=v(S)-\sum_{i\in S}v(i)$ for all $S\subseteq N$. Notice
that $ESD$ is a variant of $ED$ in which we first allocate $v(i)$ to each
player $i\in N$, and then distribute $v^{0}(N)$ among the players using $ED$%
. $ESD$ is a reasonable alternative to $ED$ for situations where individual
benefits and joint benefits are neatly separable. Let us illustrate this
with the example of Section 2 (notice that it deals with costs instead of
with benefits).

Consider again the three-storey building of Section 2 and the cost of
installing the elevator. Clearly, the cost of the machine is a joint cost,
whereas the cost due to the works to be done on each floor should be paid by
the owners of each floor. With respect to the costs of the hollow, assume
that there is a fixed cost of 10 and an individual cost of 10 for the owners
of the first floor that is incremented by 10 for the owners of the second
floor and by an additional 10 for the owners of the third floor. According
to this, the cost $c(i)$ in which each player is involved is:

\begin{itemize}
\item 50 (machine) + 10 (floor) + 40 (hollow) = 100, for the players of the
third floor,

\item 50 (machine) + 10 (floor) + 30 (hollow) = 90, for the players of the
second floor,

\item 50 (machine) + 10 (floor) + 20 (hollow) = 80, for the players of the
first floor.
\end{itemize}

Now we can compute the equal surplus division value for the game in which the
players are the apartments and $c(N)=120$ (this is what we call the ES-Dutch
rule) and the equal surplus division value for the game in which the players
are the quota units and $c(N)=120$ (this is what we call the ES-Spanish
rule). Table \ref{table3} below displays the distributions of the cost among
the apartments using both rules. Notice that these distributions are not
satisfactory because they seem to penalize too much the apartments on the
third floor, specially the ES-Spanish rule that even proposes that the
apartment on the first floor is recompensed if the elevator is installed.
The reason for this seems to be that the individual costs in this example
actually belong to the floors instead of to the players; consequently it
would be more reasonable to use a kind of two-step rule for the equal
surplus division value analogous to the two-step rule for the equal division
value introduced in Section 2. In other words, this example suggests that we
should consider the structure of a priori unions given by the floors and
distribute the costs using an extension of the equal surplus division value
to TU-games with a priori unions.

\vspace*{0.5cm}

\begin{table}[htbp]
\begin{center}
\begin{tabular}{|l|l|l|}
\hline
& Dutch rule & Spanish rule \\ \hline
3rd floor & 26.6666 26.6666 26.6666 & 613.0860 613.0860 613.0860 \\ \hline
2nd floor & 16.6666 16.6666 & 21.8100 19.6290 \\ \hline
1st floor & 6.6666 & -1760.7240 \\ \hline
\end{tabular}%
\end{center}
\caption{Distribution according to the ES-Dutch and ES-Spanish rules}
\label{table3}
\end{table}


Next we propose three alternative ways for extending the equal surplus
division value to TU-games with a priori unions. The first one divides the
value of the grand coalition in the quotient game using the equal surplus
division value and then divides the amount assigned to each union equally
among its members.

\begin{defin}
The equal surplus division value (one) for TU-games with a priori unions $%
ESD1^U$ is defined by 
\begin{equation*}
ESD1^U_{i}\left(N,v,P\right)=\frac{(v/P)(k)}{p_k}+\frac{(v/P)^0(M)}{mp_k}=%
\frac{v(P_k)}{p_k}+\frac{v(N)-\sum_{l\in M}v(P_l)}{mp_k}
\end{equation*}
for all $i\in N$ and all $(N,v,P)\in\mathcal{G}^U$ with $P=\{P_1,\dots,P_m\}$
and with $i\in P_k$.
\end{defin}

The second extension divides again the value of the grand coalition in the
quotient game using the equal surplus division value; then it distributes
the amount $\frac{v(N)-\sum_{l\in M}v(P_{l})}{m}$ equally among the players 
in each union, and the amount $v(P_k)$ giving $v(i)$ to each
player $i\in P_k$ and dividing $v(P_k)-\sum_{j\in P_k}v(j)$ equally among
the players in $P_k$.

\begin{defin}
The equal surplus division value (two) for TU-games with a priori unions $%
ESD2^{U}$ is defined by 
\begin{equation*}
ESD2_{i}^{U}\left( N,v,P\right) =v(i)+\frac{v(P_{k})-\sum_{j\in P_{k}}v(j)}{%
p_{k}}+\frac{v(N)-\sum_{l\in M}v(P_{l})}{mp_{k}}
\end{equation*}%
for all $i\in N$ and all $(N,v,P)\in \mathcal{G}^{U}$ with $P=\{P_{1},\dots
,P_{m}\}$ and with $i\in P_{k}$.
\end{defin}

Finally, the third extension assigns $v(i)$ to each player $i$ and then
divides $v^0(N)$ among the players using $ED^U$.

\begin{defin}
The equal surplus division value (three) for TU-games with a priori unions $%
ESD3^U$ is defined by 
\begin{equation*}
ESD3^U_{i}\left(N,v,P\right)=v(i)+ED^U(N,v^0,P)=v(i)+ \frac{v(N)-\sum_{j\in
N}v(j)}{mp_k}
\end{equation*}
for all $i\in N$ and all $(N,v,P)\in\mathcal{G}^U$ with $P=\{P_1,\dots,P_m\}$
and with $i\in P_k$.
\end{defin}

Now we can compute the equal surplus division values one, two and three for
the game with a priori unions in which the players are the apartments, the
unions are the floors and $c(N)=120$ (they are what we call the $ESD1^U$, $%
ESD2^U$ and $ESD3^U$-Dutch rules) and the equal surplus division values one,
two and three for the game with a priori unions in which the players are the
quota units, the unions are the floors and $c(N)=120$ (they are what we call
the $ESD1^U$, $ESD2^U$ and $ESD3^U$-Spanish rules). Tables \ref{table4}, \ref%
{table5} and \ref{table6} below display the distributions of the cost among
the apartments using these rules.\footnote{Notice that Tables 4 and 5 are identical. This is because $ESD1$ and $ESD2$ coincide when, as in this example, for each union $P_k$ and each $i,j\in P_k$, it is satisfied that $v(i)=v(j)$.} The results in Tables \ref{table4} and \ref%
{table5} seem to be more reasonable than those in Table \ref{table3}; notice
that they slightly penalize the higher floors in comparison with the results
in Table \ref{table2}. The result in Table \ref{table6} is not satisfactory
since it penalizes too much the apartments on the third floor. It shows that 
$ESD3^U$ is not an appropriate extension of $ESD$, at least for this
example; we informally discuss why in the section of concluding remarks.

\vspace*{0.5cm}

\begin{table}[htbp]
\begin{center}
\begin{tabular}{|l|l|l|}
\hline
& Dutch rule & Spanish rule \\ \hline
3rd floor & 16.6666 16.6666 16.6666 & 16.6666 16.6666 16.6666 \\ \hline
2nd floor & 20 20 & 21.0584 18.9525 \\ \hline
1st floor & 30 & 30 \\ \hline
\end{tabular}%
\end{center}
\caption{Distribution according to $ESD1^U$}
\label{table4}
\end{table}
\begin{table}[htbp]
\begin{center}
\begin{tabular}{|l|l|l|}
\hline
& Dutch rule & Spanish rule \\ \hline
3rd floor & 16.6666 16.6666 16.6666 & 16.6666 16.6666 16.6666 \\ \hline
2nd floor & 20 20 & 21.0584 18.9525 \\ \hline
1st floor & 30 & 30 \\ \hline
\end{tabular}%
\end{center}
\caption{Distribution according to $ESD2^U$}
\label{table5}
\end{table}
\begin{table}[htbp]
\begin{center}
\begin{tabular}{|l|l|l|}
\hline
& Dutch rule & Spanish rule \\ \hline
3rd floor & 51.1111 51.1111 51.1111 & 513.3333 513.3333 513.3333 \\ \hline
2nd floor & 16.6666 16.6666 & 336.8421 303.1579 \\ \hline
1st floor & -66.6666 & -2060 \\ \hline
\end{tabular}%
\end{center}
\caption{Distribution according to $ESD3^U$}
\label{table6}
\end{table}

In the remainder of this section we study $ESD1^{U}$, $ESD2^{U}$ and $%
ESD3^{U}$ from the point of view of their properties; in particular, we
provide axiomatic characterizations of these values. We start by introducing
new properties of a value $g$ for TU-games with a priori unions. Take $%
(N,v)\in \mathcal{G}$ and $i\in N$. We say that $i$ is a dummifying player
in $v$ if $v(S\cup i)=\sum_{j\in S\cup i}v(j)$ for all $S\subseteq N$. Take
now a TU-game with a priori unions $(N,v,P)\in \mathcal{G}^{U}$ where $%
P=\{P_{1},\dots ,P_{m}\}$. We say that $P_{k}$ is a dummifying union in $%
(v,P)$ if $k$ is a dummifying player in $v/P$. Dummifying players and
dummifying unions should play a relevant role in the characterizations of $%
ESD1^{U}$, $ESD2^{U}$ and $ESD3^{U}$ since a property on dummifying players
is used in Casajus and H\"{u}ttner (2014) for characterizing $ESD$. In fact
they use the following property (for $\mathcal{G}$ instead of $\mathcal{G}%
^{U}$).

\bigskip \noindent \textbf{Dummifying player property (DPP).} A value $g$
for TU-games with a priori unions satisfies the dummifying player property
if, for all $(N,v,P)\in\mathcal{G}^U$ and all $i\in N$ dummifying player in $%
v$, it holds that $g_{i}( N,v,P) =v(i)$.

\bigskip Notice that $ESD3^U$ satisfies DPP, but neither $ESD1^U$ nor $%
ESD2^U $ satisfy it. In the search of properties that $ESD1^U$ or $ESD2^U$
might satisfy, we propose the following variations of DPP and NPP.

\bigskip \noindent \textbf{Dummifying union/player property (DUPP).} A value 
$g$ for TU-games with a priori unions satisfies the dummifying union/player
property if, for all $(N,v,P)\in \mathcal{G}^{U}$ and all $P_{k}\in P$
dummifying union in $(v,P)$ with $i\in P_{k}$ being a dummifying player in $%
v_{P_{k}},$\footnote{$v_{P_{k}}$ denotes the characteristic function of the
TU-game $(P_{k},v_{P_{k}})$, where $v_{P_{k}}(S)=v(S)$ for all $S\subseteq
P_{k}$.} it holds that $g_{i}(N,v,P)=v(i)$.

\bigskip \noindent \textbf{Dummifying union/nullifying player property
(DUNPP).} A value $g$ for TU-games with a priori unions satisfies the
dummifying union/nullifying player property if, for all $(N,v,P)\in \mathcal{%
G}^{U}$ and all $P_{k}\in P$ dummifying union in $(v,P)$ with $i\in P_{k}$
being a nullifying player in $v_{P_{k}},$ it holds that $g_{i}(N,v,P)=0$.

Now we give parallel characterizations of the three extensions of $ESD$
using the properties we have introduced above.

\begin{theo}
$ESD1^U$ is the unique value for TU-games with a priori unions that
satisfies ADD, SWU, SAU and DUNPP. \label{th4}
\end{theo}

\begin{proof} It is immediate to check that $ESD1^{U}$ satisfies ADD, SWU, SAU and
	DUNPP. To prove the unicity, consider a value $g$ for TU-games with a priori
	unions that satisfies ADD, SWU, SAU and DUNPP. Take $(N,v,P)\in \mathcal{G}^{U}$ 
	with $P=\{P_1,\dots,P_m\}$ and define the TU-game $(N,v^{1})$ given by
	
	\begin{equation*}
	v^1(S) = \sum_{P_{l}\subseteq S}v(P_{l}) = \sum_{l=1}^{m}v^{P_l}(S)
	\end{equation*}
	for all $S\subseteq N$, where $v^{P_l}(S) = v(P_l)$ if $P_{l}\subseteq S$ 
	and $v^{P_l}(S) = 0$ otherwise. 
	
	Take $P_k\in P$. Since $g$ is a value for TU-games with a priori unions, then
	\begin{equation*}
	\sum_{i\in N}g_i(N,v^{P_k},P)=v^{P_k}(N)=v(P_k).
	\end{equation*}
	All unions $P_l\in P$ are dummifying unions in $(v^{P_k},P)$ and all players 
	$i\in P_l$, with $l\neq k$, are nullifying players in $(v^{P_k})_{P_l}$. By DUNPP, 
	$g_i(N,v^{P_k},P)=0$ for all $i\notin P_k$. And since all
	players in $P_k$ are indistinguishable in $v^{P_k}$, then SWU implies that, 
	for all $i\in P_k$, $g_i(N,v^{P_k},P) = \frac{v(P_k)}{p_k}$. 
	Using the additivity of $g$, for all $i\in P_k$, 
	\begin{equation}  \label{eq200}
	g_i(N,v^1,P)=\frac{v(P_k)}{p_k}.
	\end{equation}
	
	Define now $v^2 = v - v^1$ and, for all $\alpha \in \mathbb{R}$ and all non-empty 
	$T\subseteq N$, $e_{T}^{\alpha}$ by 
	$e_{T}^{\alpha }(S)=\alpha$ if $S=T$ and $e_{T}^{\alpha}(S)=0$ if $S\neq T$. 
	It is clear that $v^2=\sum_{T\subseteq N}e^{v^2(T)}_T$. If $T=N$,
	since all players in $N$ are indistinguishable in $e^{v^2(N)}_N$ and all players 
	in $M$ are indistinguishable in $e^{v^2(N)}_N/P$, SWU and SAU imply that, 
	for all $i\in P_k$,
	\begin{equation*}
	g_i(N,e^{v^2(N)}_N,P) = \frac{{v^2(N)}}{mp_k} 
	= \frac{v(N) - \sum_{l\in M} v(P_l)}{mp_k}. 
	\end{equation*}
	If $T\subset N$, consider two cases:
	
	\begin{itemize}
		\item Take $T=\cup_{l\in L}P_l$, with $\emptyset\subset L\subset M$. 
		For all $P_u\in P$, if $T\neq P_u$ then $e_T^{v^{2}(T)}(P_u)=0$ and 
		if $T=P_u$ then $e_T^{v^{2}(T)}(P_u)=v^2(P_u)=0$. Hence, it is easy to 
		see that all the unions in $M\setminus L$ are dummifying unions in 
		$(e_T^{v^{2}(T)},P)$. Also, since all players in $N\setminus T$ are 
		nullifying players in $e^{v^2(T)}_T$, DUNPP implies that 
		$g_i(N,e^{v^2(T)}_T,P)=0$ for all $i\notin T$. Notice that since all
		unions in $L$ are indistinguishable in $e^{v^2(T)}_T$, then by SAU 
		$\sum_{i\in P_k}g_i(N,e^{v^{2}(T)}_T,P)
		=\sum_{i\in P_l}g_i(N,e^{v^{2}(T)}_T,P)$ for all $k,l\in L$; 
		notice also that since
		\begin{equation*}
		\sum_{i\in T} g_i(N,e^{v^{2}(T)}_T,P) = \sum_{i\in N} g_i(N,e^{v^{2}(T)}_T,P) 
		=  e^{v^{2}(T)}_T(N) = 0
		\end{equation*}
		then $\sum_{i\in P_k}g_i(N,e^{v^{2}(T)}_T,P)=0$ for all $k\in L$. To conclude, SWU implies that $g_i(N,e^{v^{2}(T)}_T,P)=0$ for all 
		$i\in P_k$, with $k\in L$, and therefore for all $i\in N$.
		
		\item For any other $T\subset N$ that is not in the previous case, 
		the quotient game $(M,e^{v^2(T)}_T/P)$ satisfies that 
		$(e^{v^2(T)}_T/P)(R)=0$ for all $R\subseteq M$ and, thus, all the unions in 
		$P$ are indistinguishable and dummifying unions in $(e^{v^2(T)}_T,P)$. 
		If $i\notin T$, then $i$ is a nullifying player in $e^{v^2(T)}_T$ 
		and DUNPP implies that $g_i(N,e^{v^2(T)}_T,P)=0$. 
		Analogously as in the previous case, SAU and SWU imply that 
		$g_i(N,e^{v^2(T)}_T,P)=0$ for all $i\in T$.
	\end{itemize}
	
	Now ADD implies that, for all $i\in P_k$ with $P_k\in P$,  
	\begin{equation}  \label{eq201}
	g_i(N,v^2,P)=\sum_{T\subseteq N}g_i(N,e_T^{v^2(T)},P)=\frac{{v^2(N)}}{mp_k}.
	\end{equation}
	
	Finally, from (\ref{eq200}), (\ref{eq201}), ADD and $v = v^1 + v^2$ it is clear that 
	\begin{equation*}
	g(N,v,P)=ESD1^U(N,v,P).
	\end{equation*}
\end{proof}

\begin{theo}
$ESD2^U$ is the unique value for TU-games with a priori unions that
satisfies ADD, SWU, SAU and DUPP. \label{th3}
\end{theo}

\begin{proof} It is immediate to check that $ESD2^U$ satisfies ADD, SWU, SAU and
	DUPP. To prove the unicity, consider a value $g$ for TU-games with a priori
	unions that satisfies ADD, SWU, SAU and DUPP. Take $(N,v,P)\in \mathcal{G}^U$ 
	with $P=\{P_1,\dots,P_m\}$ and define $v^a$, $v^{01}$ and $v^{02}$ by:
	
	\begin{itemize}
		\item $v^a(S)=\sum_{i\in S}v(i)$,
		
		\item $v^{01}(S)=\sum_{P_{l}\subseteq S}v^0(P_{l})
		=\sum_{l=1}^{m}v^{0P_l}$(S),
		
		\item $v^{02}(S)=v^0(S)-\sum_{P_{l}\subseteq S}v^0(P_{l})$,
	\end{itemize}
	for all $S\subseteq N$, where $v^{0P_l}(S) = v^0(P_l)$ if $P_{l}\subseteq S$ 
	and $v^{0P_l}(S) = 0$ otherwise. 
	
	Since all unions are dummifying in $(v^a,P)$ and all players are dummifying
	in $v^a$, then DUPP implies that, for all $i\in N$,
	\begin{equation}  \label{eq199}
	g_i(N,v^{a},P)=v^a(i)=v(i).
	\end{equation}
	
	Take $P_k\in P$. Since $g$ is a value for TU-games with a priori unions, then
	\begin{equation*}
	\sum_{i\in N}g_i(N,v^{0P_k},P)=v^{0P_k}(N)=v^0(P_k).
	\end{equation*}
	All unions $P_l\in P$ are dummifying unions in $(v^{0P_k},P)$ and all players 
	$i\in P_l$, with $l\neq k$, are dummifying players in $(v^{0P_k})_{P_l}$. 
	By DUPP, $g_i(N,v^{0P_k},P)=v^{0P_k}(i)=0$ for all $i\notin P_k$. And
	since all players in $P_k$ are indistinguishable in $v^{0P_k}$, then 
	SWU implies that, for all $i\in P_k$, $g_i(N,v^{0P_k},P)=\frac{v^0(P_k)}{p_k}$. 
	Using ADD, for all $i\in P_k$, 
	\begin{equation}  \label{eq198}
	g_i(N,v^{01},P)=\frac{v^0(P_k)}{p_k}.
	\end{equation}
	
	Take now into account that $v^{02}=\sum_{T\subseteq N}e^{v^{02}(T)}_T$. If 
	$T=N$, since all players in $N$ are indistinguishable in $e^{v^{02}(N)}_N$ 
	and all players in $M$ are indistinguishable in $e^{v^{02}(N)}_N/P$, SWU 
	and SAU imply that, for all $i\in P_k$, 
	\begin{equation*}
	g_i(N,e^{v^{02}(N)}_N,P)=\frac{{v^{02}(N)}}{mp_k}. 
	\end{equation*}
	If $T\subset N$, consider two cases:
	
	\begin{itemize}
		\item Take $T=\cup_{l\in L}P_l$, with $\emptyset\subset L\subset M$.
		Since $e_T^{v^{02}(T)}(P_u)=0$ for all $P_u\in P$ and $(e_T^{v^{02}(T)}/P)(R)=0$ 
		for all $R\subseteq M$ with $R\cap(M\setminus L)\neq\emptyset$, 
		all the unions in $M\setminus L$ are dummifying unions in $(e_T^{v^{02}(T)},P)$. 
		Also, since all players in $N\setminus T$ are dummifying players in 
		$e^{v^{02}(T)}_T$, DUPP implies that $g_i(N,e^{v^{02}(T)}_T,P)=
		e^{v^{02}(T)}_T(i)=0 $ for all $i\notin T$. Notice that since all unions in 
		$L$ are indistinguishable in  $e^{v^{02}(T)}_T$, then by SAU 
		$\sum_{i\in P_k}g_i(N,e^{v^{02}(T)}_T,P)=\sum_{i\in P_l}g_i(N,e^{v^{02}(T)}_T,P)$ 
		for all $k,l\in L$, and notice that since
		\begin{equation*}
		\sum_{i\in T} g_i(N,e^{v^{02}(T)}_T,P) = 
		\sum_{i\in N} g_i(N,e^{v^{02}(T)}_T,P) = e^{v^{02}(T)}_T(N) = 0
		\end{equation*}
		then $\sum_{i\in 	P_k}g_i(N,e^{v^{02}(T)}_T,P)=0$ for all $k\in L$. Hence, SWU implies 
		that $g_i(N,e^{v^{02}(T)}_T,P)=0$ for all $i\in T$.
		
		\item For any other $T\subset N$ that is not in the previous case, the quotient game 
		$(M,e_T^{v^{02}(T)}/P)$ satisfies that $(e_T^{v^{02}(T)}/P)(R)=0$ for all 
		$R\subseteq M$ and, thus, all the unions in $P$ are indistinguishable and dummifying 
		unions in $(e_T^{v^{02}(T)},P)$. If $i\notin T$, then $i$ is a dummifying player in 
		$e^{v^{02}(T)}_T$ and DUPP implies that $g_i(N,e^{v^{02}(T)}_T,P)=
		e^{v^{02}(T)}_T(i)=0$. Analogously as in the previous case, SAU and SWU imply that 
		$g_i(N,e^{v^{02}(T)}_T,P)=0$ for all $i\in T$. 
	\end{itemize}
	
	Now ADD implies that, for all $i\in P_k$ with $P_k\in P$, 
	\begin{equation}  \label{eq197}
	g_i(N,v^{02},P)=\sum_{T\subseteq N}g_i(N,e_T^{v^{02}(T)},P)
	=\frac{{v^{02}(N)}}{mp_k}.
	\end{equation}
	
	Finally, from (\ref{eq199}), (\ref{eq198}), (\ref{eq197}), ADD and 
	$v = v^{a} + v^{01} + v^{02}$ it is clear that 
	\begin{equation*}
	g(N,v,P)=ESD2^U(N,v,P).
	\end{equation*}
\end{proof}

Now we provide a characterization of $ESD3^U$. In order to do it we
introduce a new property that is a weaker version of SAU.

\bigskip

\noindent \textbf{Weak symmetry among unions (WSAU).} A value $g$ for
TU-games with a priori unions satisfies weak symmetry among unions if, for
all $(N,v,P)\in\mathcal{G}^U$ with $v(j)=0$ for all $j\in N$, and for all $%
k,l\in M$ indistinguishable in $v/P$, it holds that $\sum_{i\in
P_k}g_i(N,v,P)=\sum_{i\in P_l}g_i(N,v,P)$.

\begin{theo}
$ESD3^U$ is the unique value for TU-games with a priori unions that
satisfies ADD, SWU, WSAU and DPP. \label{th5}
\end{theo}

\begin{proof} It is immediate to check that $ESD3^U$ satisfies ADD, SWU, WSAU and DPP. To
prove the unicity, consider a value $g$ for TU-games with a priori unions
that satisfies ADD, SWU, WSAU and DPP. Take now $(N,v,P)\in\mathcal{G}^U$
and $i\in P_k$ with $P_k\in P$, and define $v^{a} = v - v^0$. ADD implies that 
\begin{equation}  \label{eq899}
g_i(N,v,P)=g_i(N,v^{a},P)+g_i(N,v^0,P).
\end{equation}

Since all players are dummifying in $v^a$, then DPP implies that 
\begin{equation}  \label{eq799}
g_i(N,v^{a},P)=v^a(i)=v(i).
\end{equation}

Now, using for $(N,v^0)$ analogous arguments as those used in the proof of
Theorem \ref{th2}, it is clear that ADD, SWU, WSAU and DPP imply that
\begin{equation}  \label{eq898}
g_i(N,v^0,P)=ED_i(N,v^0,P).
\end{equation}

Finally, from (\ref{eq899}), (\ref{eq799}) and (\ref{eq898}) it is clear that 
\begin{equation*}
g(N,v,P)=ESD3^U(N,v,P).
\end{equation*}
\end{proof}

\section{Concluding remarks}

In this last section, we include some supplementary information.

\medskip a) It is immediate to prove that $ESD3^{U}$ does not satisfy SAU. Since WSAU is 
a weaker version of SAU, and $ESD3^{U}$ is characterized with ADD, SWU, WSAU
and DPP, we conclude that there does not exist a value for TU-games with a priori unions satisfying ADD, SWU,
SAU and DPP.

b) Given a value $f$ for TU-games, a \emph{coalitional }$f$ 
\emph{value} is a value $g$ for TU-games with a priori unions that
coincides with $f$ when the partition $P$ is such that each union is a
singleton. That is, if we denote by $P^{n}$ the partition $\left\{ \left\{
1\right\} ,\left\{ 2\right\} ,\ldots ,\left\{ n\right\} \right\} ,$ it holds
that $g(N,v,P^{n})=f(N,v)$. It is easy to check that $ED^{U}$ is a \emph{%
coalitional equal division value}, and $ESD1^{U}$, $ESD2^{U}$ and $ESD3^{U}$
are \emph{coalitional equal surplus division values}.

c) A value $g$ for TU-games with a priori unions satisfies 
the \emph{quotient game property} (QGP) if, for all $( N,v,P)\in\mathcal{G}^U$ 
with $P=\{P_1,\dots,P_m\}$ and for its quotient game $(M,v/P)$, it holds that 
$\sum_{i\in P_k}g_{i}\left( N,v,P\right) =g_{k}\left( M,v/P,P^{m}\right)$ 
for all $P_k\in P$. 
It is easy to check that $ED^U$, $ESD1^U$ and $ESD2^U$ satisfy QGP. 
However, $ESD3^U$ does not satisfy QGP. Maybe that is the reason why it 
does not behave in an appropriate way in the example we dealt with in 
Section \ref{sec_esdv}.

d) The properties in the theorems of this paper are independent. We prove it in the Appendix.


\section*{Acknowledgements}

\noindent This work has been supported by the ERDF, the MINECO/AEI grants
MTM2017-87197-C3-1-P, MTM2017-87197-C3-3-P, and by the Xunta de Galicia
(Grupos de Referencia Competitiva ED431C-2016-015 and ED431C-2017/38 and
Centro Singular de Investigaci\'{o}n de Galicia ED431G/01).

\section*{References}

\noindent Alonso-Meijide JM, Casas-M\'{e}ndez B, Fiestras-Janeiro G, Holler
MJ (2011). The Deegan-Packel index for simple games with a priori unions.
Quality \& Quantity 45, 425-439.\newline
\noindent Alonso-Meijide JM, Costa J, Garc\'{\i}a-Jurado I (2019). Null,
Nullifying, and Necessary Agents: Parallel Characterizations of the Banzhaf
and Shapley Values. Journal of Optimization Theory and Applications 180,
1027-1035.\newline
\noindent Alonso-Meijide JM, Fiestras-Janeiro G (2002). Modification of the
Banzhaf value for games with a coalition structure. Annals of Operations
Research 109, 213-227.\newline
\noindent Ar\'{\i}n J, Kuipers J, Vermeulen D (2003). Some characterizations
of egalitarian solutions on classes of TU-games. Mathematical Social
Sciences 46, 327-345.\newline
\noindent Aumann RJ, Dr\`{e}ze J (1974). Cooperative games with coalition
structures. International Journal of Game Theory 3, 217-237.\newline
\noindent B\'{e}al S, R\'{e}mila E, Solal P (2019). Coalitional desirability
and the equal division value. Theory and Decision 86, 95-106.\newline
\noindent Casajus A, H\"{u}ttner F (2014). Null, nullifying, or dummifying
players: The difference between the Shapley value, the equal division value,
and the equal surplus division value. Economics Letters 122, 167-169.\newline
\noindent Casas-M\'{e}ndez B, Garc\'{\i}a-Jurado I, van den Nouweland A, V%
\'{a}zquez-Brage M (2003). An extension of the $\tau $-value to games with
coalition structures. European Journal of Operational Research 148, 494-513.%
\newline
\noindent Chun Y, Park B (2012). Population solidarity, population
fair-ranking and the egalitarian value. International Journal of Game Theory
41, 255-270.\newline
\noindent Costa J (2016). A polynomial expression of the Owen value in the
maintenance cost game. Optimization 65, 797-809.\newline
\noindent Crettez B, Deloche R (2018). A law-and-economics perspective on
cost-sharing rules for a condo elevator. To appear in Review of Law \&
Economics. doi: 10.1515/rle-2016-0001.\newline
\noindent Dietzenbacher B, Borm P, Hendrickx R (2017). The procedural
egalitarian solution. Games and Economic Behavior 106, 179-187.\newline
\noindent Driessen TSH, Funaki Y (1991). Coincidence of and collinearity 
between game theoretic solutions. OR Spectrum 13, 15-30.\newline
\noindent Dutta B (1990). The egalitarian solution and reduced game
properties in convex games. International Journal of Game Theory 19, 153-169.%
\newline
\noindent Dutta B, Ray D (1989). A concept of egalitarianism under
participation constraints. Econometrica 57, 615-635.\newline
\noindent Ferri\`{e}res S (2017). Nullified equal loss property and equal
division values. Theory and Decision 83, 385-406.\newline
\noindent Hu XF, Xu GJ, Li DF (2019). The egalitarian efficient extension of
the Aumann-Dr\`{e}ze value. Journal of Optimization Theory and Applications
181, 1033-1052.\newline
\noindent Ju Y, Borm P, Ruys P (2007). The consensus value: A new solution
concept for cooperative games. Social Choice and Welfare 28, 685-703.\newline
\noindent Klijn F, Slikker M, Tijs S, Zarzuelo J (2000). The egalitarian
solution for convex games: some characterizations. Mathematical Social
Sciences 40, 111-121.\newline
\noindent Lorenzo-Freire S (2016). On new characterizations of the Owen
value. Operations Research Letters 44, 491-494.\newline
\noindent Moretti S, Patrone F (2008). Transversality of the Shapley value.
Top 16, 1-41.\newline 
\noindent Owen G (1977) Values of games with a priori unions. In: Mathematical
Economics and Game Theory (R Henn, O Moeschlin, eds.), Springer, 76-88.\newline
\noindent Saavedra-Nieves A, Garc\'{\i}a-Jurado I, Fiestras-Janeiro G
(2018). Estimation of the Owen value based on sampling. In: The Mathematics
of the Uncertain: A Tribute to Pedro Gil (E Gil, E Gil, J Gil, MA Gil,
eds.), Springer, 347-356.\newline
\noindent Selten R (1972). Equal share analysis of characteristic function
experiments. In: Contributions to Experimental Economics III. (Sauermann H,
ed.), Mohr Siebeck, 130-165.\newline
\noindent Shapley LS (1953). A value for n-person games. In: Contributions
to the Theory of Games II (HW Kuhn, AW Tucker, eds.), Princeton University
Press, 307-317.\newline
\noindent van den Brink R (2007). Null or nullifying players: the difference
between the Shapley value and equal division solutions. Journal of Economic
Theory 136, 767-775.\newline
\noindent van den Brink R, Funaki Y (2009) Axiomatizations of a class of
equal surplus sharing solutions for TU-games. Theory and Decision 67,
303-340.\newline
\noindent van den Brink R, Chun Y, Funaki Y, Park B (2016). Consistency, population
solidarity, and egalitarian solutions for TU-games. Theory and Decision 81,
427-447.\newline

\section*{Appendix}

\noindent a) Independence of the properties of Theorem \ref{th2}:

\begin{itemize}
	
		\item $\varphi _{i}=v(i)$ satisfies ADD,
	SWU, SAU and NPP, but not EFF.
	
	\item $\varphi _{i}=$ $\frac{v(P_{k})}{p_{k}}+\frac{v(N)-\sum_{l\in
			M}v(P_{l})}{mp_{k}}$ satisfies EFF, ADD, SWU and SAU, but not NPP.
	
	\item $\varphi _{i}=$ $\frac{v(N)}{n}$ satisfies EFF, ADD, SWU and NPP, but not
	SAU.
	
	\item $\varphi _{i}=\frac{2v(N)}{mp_{k}}$ if $\displaystyle i=\min_{j\in
		P_{k}}j$ or $\varphi _{i}=\frac{(p_{k}-2)v(N)}{mp_{k}(p_{k}-1)}$ if $i\in
	P_{k}$ and $i\neq \min_{j\in P_{k}}j$, satisfies EFF, ADD, SAU and NPP, but not
	SWU.
	
	\item $\varphi _{i}=\frac{2v(N)}{mp_{k}|Z_{k}|}$ if $\displaystyle i\in
	Z_{k}=\{j\in P_{k}/v(j)=\min_{z\in P_{k}}v(z)\}$, $\varphi _{i}=\frac{%
		(p_{k}-2)v(N)}{mp_{k}(p_{k}-|Z_{k}|)}$ if $i\in P_{k}\backslash Z_{k}$,
	satisfies EFF, SWU, SAU and NPP, but not ADD.
\end{itemize}

\noindent b) Independence of the properties of Theorem \ref{th4}:

\begin{itemize}
	
		\item $\varphi _{i}=v(i)$ satisfies ADD,
	SWU, SAU and DUNPP, but not EFF.
	
	\item $\varphi _{i}=\frac{v(P_{k})}{p_{k}}+\frac{v(N)-\sum_{l\in M}v(P_{l})}{n}$ satisfies EFF, ADD, SWU, DUNPP, but not SAU.
	
	\item $\varphi _{i}=\frac{v(P_{k})}{p_{k}}+\frac{2(v(N)-\sum_{l\in
			M}v(P_{l}))}{mp_{k}}$ if $\displaystyle i=\min_{j\in P_{k}}j$ or $\varphi
	_{i}=\frac{v(P_{k})}{p_{k}}+\frac{(p_{k}-2)(v(N)-\sum_{l\in
			M}v(P_{l}))}{%
		mp_{k}(p_{k}-1)}$ if $i\in P_{k}$ and $i\neq \min_{j\in P_{k}}j$, satisfies
	EFF, ADD, SAU and DUNPP, but not SWU.
	
	\item $\varphi _{i}=\frac{v(P_{k})}{p_{k}}+\frac{2(v(N)-\sum_{l\in
			M}v(P_{l}))}{mp_{k}|Z_{k}|}$ if $\displaystyle i\in Z_{k}=\{j\in
	P_{k}/v(j)=\min_{z\in P_{k}}v(z)\}$, $\varphi _{i}=\frac{v(P_{k})}{p_{k}}+%
	\frac{(p_{k}-2)(v(N)-\sum_{l\in
			M}v(P_{l}))}{mp_{k}(p_{k}-|Z_{k}|)}$ if $%
	i\in P_{k}\backslash Z_{k}$, satisfies EFF, SWU, SAU and DUNPP, but not ADD.
	
	\item $\varphi _{i}=v(i)+\frac{v(P_{k})-\sum_{j\in P_{k}}v(j)}{p_{k}}+\frac{%
		v(N)-\sum_{l\in
			M}v(P_{l})}{mp_{k}}$ satisfies EFF, ADD, SWU and SAU, but not
	DUNPP.
\end{itemize}

\noindent c) Independence of the properties of Theorem \ref{th3}:

\begin{itemize}
	
		\item $\varphi _{i}=v(i)$ satisfies ADD,
	SWU, SAU and DUPP, but not EFF.
	
	\item $\varphi _{i}=v(i)+\frac{v(P_{k})-\sum_{j\in P_{k}}v(j)}{p_{k}}+\frac{%
		v(N)-\sum_{l\in
			M}v(P_{l})}{n}$ satisfies EFF, ADD, SWU and DUPP, but not SAU.
	
	\item $\varphi _{i}=v(i)+\frac{v(P_{k})-\sum_{j\in P_{k}}v(j)}{p_{k}}+\frac{%
		2(v(N)-\sum_{l\in
			M}v(P_{l}))}{mp_{k}}$ if $\displaystyle i=\min_{j\in
		P_{k}}j$ or $\varphi _{i}=v(i)+\frac{v(P_{k})-\sum_{j\in P_{k}}v(j)}{p_{k}}+%
	\frac{(p_{k}-2)(v(N)-\sum_{l\in
			M}v(P_{l}))}{mp_{k}(p_{k}-1)}$ if $i\in
	P_{k}$ and $i\neq \min_{j\in P_{k}}j$, satisfies EFF, ADD, SAU and DUPP, but not
	SWU.
	
	\item $\varphi _{i}=v(i)+\frac{v(P_{k})-\sum_{j\in P_{k}}v(j)}{p_{k}}+\frac{%
		2(v(N)-\sum_{l\in
			M}v(P_{l}))}{mp_{k}|Z_{k}|}$ if $\displaystyle i\in
	Z_{k}=\{j\in P_{k}/v(j)=\min_{z\in P_{k}}v(z)\}$, $\varphi _{i}=v(i)+\frac{%
		v(P_{k})-\sum_{j\in P_{k}}v(j)}{p_{k}}+\frac{(p_{k}-2)(v(N)-\sum_{l\in
			M}v(P_{l}))}{mp_{k}(p_{k}-|Z_{k}|)}$ if $i\in P_{k}\backslash Z_{k}$,
	satisfies EFF, SWU, SAU and DUNPP, but not ADD.
	
	\item $\varphi _{i}=\frac{v(P_{k})}{p_{k}}+\frac{v(N)-\sum_{l\in
			M}v(P_{l})}{mp_{k}}$ satisfies EFF, ADD, SWU and SAU, but not DUPP.
\end{itemize}

\noindent d) Independence of the properties of Theorem \ref{th5}:

\begin{itemize}
	
		\item $\varphi _{i}=v(i)$ satisfies ADD,
	SWU, WSAU and DPP, but not EFF.
	
	\item $\varphi _{i}=v(i)+\frac{v(N)-\sum_{j\in N}v(j)}{n}$ satisfies EFF, ADD,
	SWU and DPP, but not WSAU.
	
	\item $\varphi _{i}=v(i)+\frac{2(v(N)-\sum_{j\in N}v(j))}{mp_{k}}$ if $%
	\displaystyle i=\min_{j\in P_{k}}j$ or $\varphi _{i}=v(i)+\frac{%
		(p_{k}-2)(v(N)-\sum_{j\in N}v(j))}{mp_{k}(p_{k}-1)}$ if $i\in P_{k}$ and $%
	i\neq \min_{j\in P_{k}}j$, satisfies EFF, ADD, WSAU and DPP, but not SWU.
	
	\item $\varphi _{i}=v(i)+\frac{2(v(N)-\sum_{j\in N}v(j))}{mp_{k}|Z_{k}|}$ if $%
	\displaystyle i\in Z_{k}=\{j\in P_{k}/v(j)=\min_{z\in P_{k}}v(z)\}$, $%
	\varphi _{i}=v(i)+\frac{(p_{k}-2)(v(N)-\sum_{j\in N}v(j))}{%
		mp_{k}(p_{k}-|Z_{k}|)}$ if $i\in P_{k}\backslash Z_{k}$, satisfies EFF, SWU, WSAU
	and DPP, but not ADD.
	
	\item $\varphi _{i}=\frac{v(P_{k})}{p_{k}}+\frac{v(N)-\sum_{l\in
			M}v(P_{l})}{mp_{k}}$ satisfies EFF, ADD, SWU, WSAU but not DPP.
\end{itemize}

\end{document}